\documentclass{amsart}
\usepackage[utf8]{inputenc}
\usepackage[margin=1in]{geometry}
\usepackage{amsfonts, amssymb, amsthm, amsmath, mathtools, graphicx, hyperref, caption, comment, color}
\usepackage{comment,tikz,epstopdf}
\usepackage{graphics}
\usepackage[normalem]{ulem}
\usepackage[none]{hyphenat}

\usepackage[bottom]{footmisc}

\hypersetup{
    colorlinks=true,
    linkcolor=blue,
    filecolor=magenta,      
    urlcolor=cyan,
    pdftitle={Overleaf Example},
    pdfpagemode=FullScreen,
    }

\theoremstyle{definition}

\newtheorem{theorem}{Theorem}[section]
\newtheorem{corollary}[theorem]{Corollary}
\newtheorem{lemma}[theorem]{Lemma}
\newtheorem{remark}{Remark}

\newtheorem{proposition}[theorem]{Proposition}
\newtheorem{example}[theorem]{Example}

\title{A Construction of Optimal Quasi-cyclic Locally Recoverable Codes using Constituent Codes}

\author{Gustavo Terra Bastos}
\address{Gustavo Terra Bastos. Department of Mathematics and Statistics, Federal University of S\~ao Jo\~ao del-Rei, 170 Frei Orlando Sq, Sao Joao del Rei, State of Minas Gerais, 36307-352, Brazil } \email{{gtbastos@ufsj.edu.br}}

\author{Angelynn \'Alvarez}
\address{Angelynn \'Alvarez. Department of Mathematics, Embry-Riddle Aeronautical University, 3700 Willow Creek Rd., Prescott, AZ 86301, United States} \email{{alvara44@erau.edu}}

\author{Zachary Flores}
\address{Zachary Flores. Two Six Technologies, 901 N Stuart St, Arlington, VA 22203, United States}
\email{{floresza03@gmail.com}}

\author{Adriana Salerno}
\address{Adriana Salerno. Department of Mathematics, Bates College, 3 Andrews Rd, Lewiston, ME, 04240, United States} \email{{asalerno@bates.edu}}

\date{}

\thanks{}


\begin{document}

\begin{abstract} A locally recoverable code of locality $r$ over $\mathbb{F}_{q}$ is a code where every coordinate of a codeword can be recovered using the values of at most $r$ other coordinates of that codeword. Locally recoverable codes are efficient at restoring corrupted messages and data which make them highly applicable to distributed storage systems. Quasi-cyclic codes of length $n=m\ell$ and index $\ell$ are linear codes that are invariant under cyclic shifts by $\ell$ places. 
In this paper, we decompose quasi-cyclic locally recoverable codes into a sum of constituent codes where each constituent code is a linear code over a field extension of $\mathbb{F}_q$. Using these constituent codes with set parameters, we propose conditions which ensure the existence of almost optimal and optimal quasi-cyclic locally recoverable codes with increased dimension and code length.
\end{abstract}

\maketitle

\section{Introduction}\label{introduction}
Cloud and distributed storage systems have reached a massive scale, making data reliability against storage failures crucial. One way to repair a storage failure is to use error-correcting codes that efficiently recover lost data with minimum overhead. One class of error-correcting codes, known as \textit{locally recoverable codes} (LRCs), can repair any erasure of a node that occurred during transmission by accessing other nodes in the code. 
LRCs have caught the attention of the information theory community due to their applications to storage systems and high performance for restoring corrupted messages and data.

 More precisely, a linear error correcting code $C$ is an $\mathbb{F}_q$-subspace of $\mathbb{F}_q^n$.  We say that $n$ is the \emph{code length} of $C$ and we let $k$ denote the \emph{dimension} of $C
$ as an $\mathbb{F}_q$-vector space. The minimum number of nonzero coordinates among the nonzero elements of $C$ is called the \emph{minimum distance} $d=d(C)$ of the code. A linear error correcting code with these parameters is called an $[n,k,d]$ code. The \textit{information rate} of $C$ is the ratio $\frac{k}{n}$ and a  rate close to $1$ implies a reduction of the amount of redundancy and overhead.


A code $C$ over $\mathbb{F}_q$ is a \textit{locally recoverable code} (LRC) if for every \textit{codeword} $\textbf{c}\in  C$, any coordinate of $\textbf{c}$ can be expressed as a function of at most $r$ other coordinates of $\textbf{c}$, excluding the original coordinate. Namely, if the $i$-th coordinate of $\textbf{c}$ is lost, it can be computed by accessing at most $r$ other coordinates in $\textbf{c}$. The \textit{locality of the $i$-th coordinate} is denoted by $r_i$ and the set of coordinates used to recover the lost coordinate is called its \textit{recovery set}. The \textit{locality of $C$} is defined as $r=\max_{1\leq i\leq n} r_i$.  It is worth noting that the recovery sets of different coordinates are disjoint.

 The minimum distance of a code determines how many errors a code can detect and correct, and codes with small locality and large dimension can reconstruct any data erasures in distributed storage systems with minimal overhead. Thus, a large minimum distance, small locality, and large dimension are desired. The relationship between the parameters $d$, $n$, $k$, and $r$ is described by the Singleton-type bound
 \begin{equation}\label{singletontypebound}
     d \leq n - k - \bigg\lceil \frac{k}{r} \bigg\rceil + 2 =: d_S (C).
 \end{equation}
 A code that achieves equality in (\ref{singletontypebound}) is called an \textit{optimal} LRC\footnote{
In the particular case where the parameters of a code yield $d_S (C)\leq 0$, we assume that the respective code does not exist.
}. Optimal codes  allow for the detection
and correction of the maximum number of errors and thus are the types of code we aim to construct.

The groundbreaking paper~\cite{TamoBarg} by Tamo and Barg started a new age in the study of LRCs by using Reed-Solomon code constructions to construct LRCs over finite fields which are optimal, meaning that they allowed for the detection and correction of the maximum number of errors. The authors addressed the construction of LRCs from cyclic codes in~\cite{tamo-barg-goparaju-calderbank} and also provided constructions of cyclic and geometric algebraic codes from the perspective of LRCs in~\cite{TamoBargcycliccodes, TamoBargalgcurves}.

The locality of cyclic codes have been discussed frequently in literature. Goparaju and Calderbank~\cite{Goparaju} published one of the first papers dealing with the locality in cyclic codes. In~\cite{Huang}, the localities of some classical cyclic code families like Hamming, Golay, and BCH have been studied in this paper. In~\cite{luo}, the authors present families of optimal LR cyclic codes with minimum distances $3$ and $4$ whose length is not restricted to other parameters. In~\cite{Pan}, the locality of cyclic codes obtained by cyclotomic polynomials is also studied, where optimal codes are introduced by a different metric than that one given by Singleton bound. 

This paper is concerned with the quasi-cyclic (QC) codes, which are a natural generalization of cyclic codes. By ~\cite{longcodes}, it is known that QC codes are asymptotically good codes, have connections with the family of convolutional codes as described in \cite{SolomonTilborg}, and have an interesting algebraic structure that has been studied in \cite{seguin} and \cite{lally}. By the works of Ling and Solé~\cite{lingsole} and G\"uneri and \"Ozbudak in~\cite{cem}, we know that QC codes in $\mathbb{F}_q ^{m\ell}$ can be represented using  \textit{constituent codes}, which are linear codes obtained via the Chinese Remainder Theorem, by projecting QC codes over finite extension fields described by irreducible factors of $x^m -1$, where $\gcd(m,q)=1$. From the minimum distances of these constituent codes, the authors in \cite{cem} presented a lower bound of the minimum distance of the respective QC code which is known to be sharper than some previous lower bounds of minimum distances available in the literature.

\par
In a recent paper~\cite{RB}, Rajput and Bhaintwal found an upper bound for the locality of QC LRCs using generator matrices of these codes in systematic form, namely where the first block is the identity matrix and the other blocks consist of smaller circulant matrices. Some particular results for 1-generator case are also provided. 

\par
In this paper, we propose conditions on the code parameters to ensure the existence of (almost) optimal QC LRCs  using constituent codes as described in \cite{cem}. Moreover, we show that our code construction can yield QC LRCs with large dimension and code length. The paper is organized as follows. In Section 2, we give background results related to QC LRCs and their constituent codes from \cite{cem}. 
In Section 3, we state previously known results on the locality of a QC LRC and describe how we can use constituent codes to determine the locality of a QC LRC. 
 In Section 4, we describe our construction of (almost) optimal QC LRCs.
We also present our main theorem that describes a descending chain of Singleton-type bounds for the minimum distance of the QC LRCs in our code construction, and then discuss some examples. 


\subsection*{Acknowledgements}
This project was initiated at the Latinx Mathematicians Research Community (LMRC) kick-off workshop organized by the American Institute of Mathematics (AIM) in Summer 2021. We give  thanks to the LMRC organizers, Jes\'us A. De Loera, Pamela E. Harris, and Anthony V\'{a}rilly-Alvarado and to AIM and NSF
for funding the LMRC. We also thank the Simons Laufer Mathematical Institute (formerly Mathematical Sciences Research Institute) for their financial support that allowed us to continue our collaboration during the 2023 Summer Research in Mathematics program. This material is based upon work supported by and while the fourth author was serving at the National Science Foundation. Any opinion, findings, and conclusions or recommendations
expressed in this material are those of the authors and do not necessarily reflect the views of the
National Science Foundation.

\section{QC LRCs and Constituent Codes}\label{preliminaries}

Let $\mathbb{F}_q$ be a finite field with $q$ elements, and $m$ and $q$ positive integers such that $\gcd(m, q)=1$.  Let $\beta$ be a primitive $m$-th root of unity over $\mathbb{F}_q$. Fix an element  $\alpha \in \mathbb{F}_{q^{\ell}}$ such that $\mathbb{F}_{q^\ell} = \mathbb{F}_{q}(\alpha)$ so that $\{1, \alpha, ...,\alpha^{\ell-1}\}$ forms an $\mathbb{F}_{q}$-basis for $\mathbb{F}_{q^\ell}$.

A \textit{quasi-cyclic} (QC) code of index $\ell$ is a code $C \subset \mathbb{F}_{q}^{m\ell}$ where for all $\textbf{c} \in C$, $T^\ell(\textbf{c}) \in C$, where $T$ is the standard shift operator on $\mathbb{F}_{q}^{m\ell}$ and $\ell$ is the smallest positive integer with this property. 
Note that a QC code becomes a cyclic code if $\ell=1$. If $C$ is a QC code in $\mathbb{F}_{q}^{m\ell}$, then any codeword $\textbf{c}$
%
can be represented as  an $m \times \ell$ array
\begin{equation}\label{codeword}
\textbf{c}=\begin{pmatrix}
c_{00}& c_{01}&...&c_{0,\ell-1} \\ 
c_{10}& c_{11}&...&c_{1,\ell-1}  \\ 
\vdots&\vdots&\ddots&\vdots\\ 
c_{m-1,0}& c_{m-1,1}&...&c_{m-1,\ell-1} 
\end{pmatrix}
 \end{equation}
where the action of the operator $T^{\ell}$ is equivalent to a row shift in~\eqref{codeword}. The commonly-used polynomial representation of codewords of a QC code $C$ is discussed in~\cite{lally, lingsole}. Additionally, one can also represent every codeword $\textbf{c} \in C$ as a vector by flattening the array in (\ref{codeword}) by the columns, namely:
\begin{equation}\label{phi}
\begin{pmatrix}
c_{00}&...&c_{0,\ell-1} \\ 
c_{10}& ...&c_{1,\ell-1}  \\ 
\vdots&\ddots&\vdots\\ 
c_{m-1,0}&...&c_{m-1,\ell-1} 
\end{pmatrix} \longmapsto (c_{00}, \ldots, c_{m-1,0},\ldots ,c_{0,\ell-1},\ldots, c_{m-1,\ell-1}).
\end{equation}

 From this point on, we will move freely between the two representations pointed out in ~\eqref{codeword} and~\eqref{phi} for codewords of a QC code $C$ of index $\ell$. The choice of representation of a codeword in $C$ is helpful in determining the locality of $C$ based on the locality of a particular cyclic code associated to $C$. This relationship will be described in detail in Section~\ref{locality}.

Consider the ring $R:=\mathbb{F}_{q}[x] / \langle x^m-1 \rangle$ where $\mathrm{gcd}(m,q)=1$. To any codeword $\textbf{c}$ as in (\ref{codeword}), we associate an element of $R^{\ell}$:
 \begin{equation}
     \vec{\textbf{c}}(x) := \left( c_0(x), c_1(x),...,c_{\ell-1}(x)\right) \in R^{\ell}
 \end{equation}
 where for each $0 \leq j \leq \ell-1$,
 \begin{equation}
     c_{j}(x):=c_{0j}+c_{1j}x+c_{2j}x^2+\cdots+c_{m-1,j}x^{m-1} \in R.
 \end{equation}
 We note that the map $\phi: \mathbb{F}_{q}^{m\ell} \to R^\ell$ defined by $\phi(\textbf{c})=\vec{\textbf{c}}(x)$ is an $R$-module isomorphism. Hence, any QC LRC over $\mathbb{F}_{q}$ of length $m\ell$ and index $\ell$ can be viewed as an $R$-submodule of $R^\ell$. 

 We factor $x^m-1$ into irreducible polynomials in $\mathbb{F}_q[x]$ as follows:
\begin{equation}\label{factorization}
    x^m-1 = b_1(x)b_2(x)\cdots b_s(x)b_{s+1}(x)
\end{equation}
where we are using the notation from~\cite{cem} and $b_{s+1}(x)=x-1$. Note that this factorization has no repeating factors and, by the Chinese Remainder Theorem, we have the following ring isomorphism:  
\begin{equation}\label{eq:crtiso} 
R \cong \bigoplus_{i=1}^{s+1} \mathbb{F}_q[x]/\langle b_i(x)\rangle.
\end{equation}
Since each $b_i(x)$ divides $x^m-1$, all roots of $b_i(x)$ are a power of a primitive $m$-th root of unity $\beta$. Thus, $b_i(\beta^{u_i}) = 0$ for some nonnegative integer $u_i$ for $1 \leq i \leq s+1$. We note that $b_{s+1}(1) = 0$, so that $u_{s+1} = 0$.  As each $b_i(x)$ is irreducible, $\mathbb{F}_q[x]/\langle b_i(x) \rangle$ is a finite field extension of $\mathbb{F}_q$, and we write 
\begin{equation}
    \mathbb{F}_q[x]/\langle b_i(x) \rangle =\mathbb{F}_{q}(\beta^{u_i})
    \end{equation}
so that $[\mathbb{F}_{q}(\beta^{u_i}) : \mathbb{F}_{q}]=\deg(b_i(x))$ for $1 \leq i \leq s+1$. We then have an isomorphism of rings
\begin{equation}\label{eq:fieldsds} 
R \cong \mathbb{F}_q(\beta^{u_1}) \oplus \mathbb{F}_q(\beta^{u_2}) \oplus \cdots \oplus \mathbb{F}_q(\beta^{u_s}) \oplus \mathbb{F}_q
\end{equation}
where $a(x) \mapsto (a(\beta^{u_1}),\ldots, a(\beta^{u_s}), a(1))$ for any $a(x) \in R$. This yields an isomorphism $\iota$ where
\begin{equation}\label{eq:rellds} 
R^\ell \stackrel{\iota}{\cong} \mathbb{F}_q^\ell(\beta^{u_1}) \oplus \mathbb{F}_q^\ell(\beta^{u_2}) \oplus \cdots \oplus \mathbb{F}_q^\ell(\beta^{u_s}) \oplus \mathbb{F}_q^\ell
\end{equation}
and any QC code $C \subseteq R^\ell$ can be viewed as an $\left( \mathbb{F}_q(\beta^{u_1})  \oplus \cdots \oplus \mathbb{F}_q(\beta^{u_s}) \oplus \mathbb{F}_q \right)$\\-submodule of $\mathbb{F}_q^\ell(\beta^{u_1}) \oplus \mathbb{F}_q^\ell(\beta^{u_2}) \oplus \cdots \oplus \mathbb{F}_q^\ell(\beta^{u_s}) \oplus \mathbb{F}_q^\ell$ with decomposition
\begin{equation}\label{eq:constituent} 
C \cong C_1 \oplus \cdots \oplus C_s \oplus C_{s+1},
\end{equation}
where each $C_i$ is a linear code of length $\ell$ over $\mathbb{F}_{q}(\beta^{u_i})$ for each $1 \leq i \leq s+1$. We call the codes $C_i$ the \textit{constituent codes} of $C$  which will play a fundamental role in the construction of (almost) optimal QC LRCs in Section~\ref{examples}.  The authors in \cite{cem} provided the following lemma that allows us to describe each constituent code as a $\mathbb{F}_{q} (\beta^{u_i})$-vector space of $\mathbb{F}_{q} ^{\ell}(\beta^{u_i})$ for each $1 \leq i \leq s+1$.
\begin{lemma}\cite[Lemma 3.1]{cem}\label{span_of_constituents}
Let $C \subset R^\ell$ be generated as an $R$-module by
$$\{(a_{1,1}(x),...,a_{1,\ell}(x)),\ldots (a_{t,1}(x),...,a_{t,\ell}(x))\} \subset R^\ell.$$
Then 
\begin{itemize}\item[(i)] $
    C_j  = \mathrm{Span}_{\mathbb{F}_{q}(\beta^{u_j})} \{ (a_{i,1}(\beta^{u_j}),\ldots,a_{i,\ell}(\beta^{u_j})) \; \mid 1 \leq i \leq t \}, \, \text{for } 1 \leq j \leq s, C_{s+1}  = \mathrm{Span}_{\mathbb{F}_{q}} \{ (a_{i,1}(1),\ldots,a_{i,\ell}(1) \; \mid 1 \leq i \leq t \}$.

\item[(ii)] For any $j \in \{1,...,s+1\}, C_j=\{0\}$ if and only if $b_{j}(x)$ divides $a_{i,k}(x)$ in $\mathbb{F}_q[x]$ for all $i=1,...,t$ and $k=1,...,\ell$.
\end{itemize}
\end{lemma}
The following dimension count will be used frequently throughout this paper, so we record it as a lemma.  
%



\begin{lemma}\label{dimmension_lemma}
Let $C\subset \mathbb{F}_q ^{m\ell}$ be a $[m\ell,k]$ QC code and let $C_{j}\subseteq\mathbb{F}_q ^\ell (\beta^{u_{j}})$ be its respective constituent $[\ell, k_{j}]$ linear codes.  Reordering if necessary, suppose the nontrivial constituent codes of $C$ are given by $C_1,\ldots ,C_h$, where $h\leq s+1$. Then, 
\begin{equation}\label{eq: dimC}
k = \sum_{j=1}^h k_j \deg(b_j(x)).
\end{equation}
\end{lemma}

\begin{proof}
Since $\mathbb{F}^{\ell} _q(\beta^{u_j}) \cong (\mathbb{F}_q[x]/(b_j(x)))^\ell,$
we have 
\begin{equation*}
C_j \cong (\mathbb{F}_q[x]/(b_j(x)))^{k_j}.  
\end{equation*}
Then $\dim_{\mathbb{F}_q }(C_j) = k_j\cdot \deg(b_j(x))$. Summing over the dimensions of the constituent codes gives (\ref{eq: dimC}).  
\end{proof}

In~\cite{cem}, the authors construct a lower bound for the minimum distance for QC codes with arbitrarily many constituent codes. This lower bound will be used extensively in Section~\ref{examples} and is notably sharper than the bound found in \cite{lally}.
For sake of brevity, we only include the main ideas of the lower bound that are used in this paper as the full details can be found in \cite{cem}. The bound is constructed as follows: Let  $\{a_1,...,a_h \}$ be a subset of $\{1,...,s+1\}$. Let $C$ be a QC code with $h$ nonzero constituent codes, i.e.,
\begin{equation}\label{r_constituents}
C \cong C_1 \oplus C_2 \oplus \cdots \oplus C_{h},
\end{equation}
where for each $1 \leq i \leq h$,
\begin{equation}\label{r_constituents_linear}
C_{i}\subset \mathbb{F}_{q}^\ell (\beta^{u_{a_i}})
\end{equation}
is a linear code over $\mathbb{F}_{q}(\beta^{u_{a_i}})$ of length $\ell$. Assume that the minimum distance of the constituent codes satisfy
\begin{equation}\label{descending_minimum_distance}
    d(C_1) \geq d(C_2) \geq \cdots \geq d(C_h).
\end{equation}
The lower bound in~\cite{cem} utilizes the following theorem (\cite[Theorem 5.1]{lingsole}) which formulates the trace representation for QC codes.

\begin{theorem} (Trace Representation of QC Codes) \label{tracerep_theorem}
Assume that $x^m -1$ factors into irreducibles as in~\eqref{factorization}. Let $U_i$ denote the $q$-cyclotomic coset$\mod m$ corresponding to $\mathbb{F}_q (\beta^{u_i})$ for all $i=1,\ldots,  s+1$ (i.e., corresponding to the powers of $\beta$ among the roots of $b_i(x)$).  Fix representatives $u_i\in U_i$ from each cyclotomic coset.  Let $C_i$ over $\mathbb{F}_q (\beta^{u_i})$ be a linear code of length $\ell$ for all $i$.  For codewords $\lambda_i\in C_i$ and for each $0\leq g\leq m-1$, let 

\begin{equation}
c_g := c_g(\lambda_1,\ldots, \lambda_{s+1}) := \sum_{i=1}^{s+1} \text{Tr}_{\mathbb{F}_q (\beta^{u_i})/\mathbb{F}_q}(\lambda_i \beta^{-gu_i})
\end{equation}
where the traces are applied to vectors coordinatewise so that $c_g$ is a vector of length $\ell$ over $\mathbb{F}_q$ for all $0\leq g\leq m-1$. Then the code
\begin{eqnarray}
C &=& \Bigg\{(c_0(\lambda_1,\ldots, \lambda_{s+1}),\ldots, c_{m-1}(\lambda_1,\ldots, \lambda_{s+1}))\nonumber \\
&=& \begin{pmatrix} c_0(\lambda_1,\ldots, \lambda_{s+1})\\
\vdots\\
c_{m-1}(\lambda_1,\ldots, \lambda_{s+1})\end{pmatrix} : \lambda_i \in C_i \Bigg\}
\end{eqnarray}
is a QC code over $\mathbb{F}_q$ of length $m\ell$ and index $\ell$. Conversely, every QC code of length $m\ell$ and index $\ell$ is obtained through this construction.
\end{theorem}

%

Note that in Theorem \ref{tracerep_theorem}, the trace-formulas are taken in different finite extension fields defined by the irreducible factors of $x^m -1 \in \mathbb{F}_q [x]$. However, by ~\cite[Lemma 4.1]{cem}, we can obtain trace-formulas from an unique finite extension field $\mathbb{F}_{q^w}$ such that $\mathbb{F}_{q} (\beta^{u_{a_i}}) \subset \mathbb{F}_{q^w}$, for  $1\leq i\leq h$, provided that there exist elements $w_i \in \mathbb{F}_{q} (\beta^{u_{a_i}})$ where $\mathrm{Tr}_{ \mathbb{F}_{q^w} / \mathbb{F}_{q} (\beta^{u_{a_i}})} (w_i) =1$. Then by~\cite[Proposition 4.6-(ii)]{cem}, the columns of any codeword $\textbf{c} \in C$ lie in a $q$-ary cyclic code $D:=D_{1,2,..,h}$ of length $m$, obtained by the trace representation of QC codes, whose dual's basic zero set is
$$BZ(D^{\perp})=\{ \beta^{-u_{a_{i}}} \mid i=1,2,...,h\}.$$ We call $D$ the \emph{cyclic code associated to $C$}, which will be useful in Section \ref{locality} when describing the locality of QC LRCs.
%
Let $I=\{i_1,i_2,...,i_t\}$ be a nonempty subset of $\{1,2,...,h\}$ where $$1 \leq i_1 < i_2 < \cdots < i_t \leq h.$$ Let $D_{I}=D_{i_1,...,i_t} \subset D$ be the $q$-ary cyclic code of length $m$ whose dual's basic zero set is
\begin{equation}\label{bz}
BZ(D_{I}^{\perp})=\{ \beta^{-u_{a_i}} \mid i \in I\}.
\end{equation}
Define the following quantity:
\begin{align*}
R_I &= R_{i_1,...,i_t}\\
& := \begin{cases}
d(C_{i_1})d(D_{i_1}) &, \text{if $t=1$} \\
\\
\left(d(C_{i_1})-d(C_{i_2})\right)d(D_{i_1})+\left(d(C_{i_2})-d(C_{i_3})\right)d(D_{i_1,i_2})+ \\
\hspace*{4.5cm}\vdots &, \text{if $t\geq 2$}\\
+ \left(d(C_{i_t-1})-d(C_{i_t})\right)d(D_{i_1,i_2,...i_{t-1}})+d(C_{i_t})d(D_{i_1,i_2,...,i_t}).
\end{cases}
\end{align*}
Using the $R_I$, we have the following bound for the minimum distance of a QC code with $h$ constituent codes.

\begin{theorem}\cite[Theorem 4.8]{cem}\label{GO_Bound_General} Let $C$ be a QC code as in (\ref{r_constituents}) and (\ref{r_constituents_linear}) with the assumption (\ref{descending_minimum_distance}). Then we have
\begin{equation}\label{GO_Bound}
    d(C) \geq  d_{\text{GO}}(C):= \min\{R_{h}, R_{h-1,h}, R_{h-2,h-1,h},...,R_{1,2,...,h} \} .
\end{equation}
\end{theorem}

The examples in Section~\ref{examples} make use of Theorem~\ref{GO_Bound_General} in the special cases of two and three constituent codes. In particular, if the number of nonzero constituent codes is two, the inequality in (\ref{GO_Bound}) becomes
\begin{equation}\label{GO_Bound_2constituents}
    d(C) \geq \min\{ d(C_2)d(D_2), \left(d(C_1)-d(C_2)\right)d(D_1)+d(C_2)d(D)\},
\end{equation}
where $D=D_{1,2}$. If the number of nonzero constituent codes is three, the inequality in  (\ref{GO_Bound}) becomes
\begin{equation}\label{GO_Bound_2constituents1}
d(C) \geq  \min\{R_3 , R_{2,3} , R_{1,2,3} \}, 
\end{equation}
where 
\begin{eqnarray*}
R_3 &=& d(C_3) d(D_3)\\
R_{2,3}&=&(d(C_2) - d(C_3) )d(D_2) + d(C_3) d(D_{2,3})\\
R_{1,2,3} &=& (d(C_1) - d(C_2) )d(D_1) + (d(C_2) -d(C_3) )d(D_{1,3}) + d(C_3) d(D),
\end{eqnarray*}
and $D=D_{1,2,3}$. It is worth noting that increasing the number of finite extension fields and, consequently, the number of non-zero constituent codes, requires a  higher computational complexity to compute the lower bound to the minimum distance provided by Theorem~\ref{GO_Bound_General}. 

\section{The Locality of QC LRCs}\label{locality} 

In this section, we provide an upper bound to the locality of an $m\ell$-length QC LRC $C$ using its cyclic code $D$ as described in Section~\ref{preliminaries}. We first state some important results on the locality of cyclic codes that will be used throughout this paper.


\begin{lemma}\cite[Lemma 2]{luo}\label{dualcode_locality}
Let $C$ be a $q$-ary transitive linear code with dual minimum distance $d^{\perp}=d(C^\perp)$. If $d^{\perp} \leq r+1$, then $C$ has locality $r$.
\end{lemma}

\begin{remark}\label{rem1}
In particular, if $C$ is a cyclic code, then $r=d^{\perp} -1$. 
\end{remark}
While QC codes are generalizations of the familiar cyclic codes, 
 the authors in  \cite{GO2012, Lim, Shi-Zhang-Sole16} showed the following:
 \begin{theorem}\cite[Theorem 1]{Shi-Zhang-Sole16}\label{Shi-Zhang-Sole16}
A quasi-cyclic code $C$ of length $m\ell $ and index $\ell$ with cyclic constituent codes is equivalent to a cyclic code.     
 \end{theorem}
%


Theorem~\ref{Shi-Zhang-Sole16} shows the importance of studying the locality of QC codes using constituent codes since, if the constituent codes are cyclic, then we can gain a better understanding of the locality of QC codes and, up to equivalence, more general cyclic codes.

To the best of our knowledge,~\cite{RB} is the first paper that provides new results on the locality of QC codes. In~\cite{RB}, the generator matrix of a QC code $C$ is described as a block matrix in standard form comprised of the identity matrix and circulant matrices. Using this generator matrix, an upper bound of the locality of the QC code is  obtained by using the weights of the polynomials associated to these circulant matrices. (See~\cite[Theorem 1]{RB}.) In order to present an alternative way to describe the locality of QC codes, we will replace the notation used in the results in~\cite{RB} with the notation used in this paper.



\begin{theorem}\cite[Theorem 2]{RB}\label{locRB}
Let $C$ be an $[m\ell,k]$ QC code. If the locality of any one coordinate in the block $\left.\left(m(i-1),mi\right.\right]$, $1\leq i\leq \ell$, is $r_i$ then all coordinates in this block have locality $r_i$.    
\end{theorem}

\begin{corollary}\cite[Corollary 1]{RB}\label{corlocRB}
The locality of an $[m\ell, k]$ QC code is 
$
r=\max_{1 \leq i \leq \ell} r_i.
$
\end{corollary}

Based on  Corollary~\ref{corlocRB}, the locality of a QC code $C$ may be obtained as the maximum of the localities provided by any column of $\textbf{c}\in C$. Notice that in this context, we are transitioning between the array and vector representations of $\textbf{c}\in C$ as described in ~\eqref{phi}. In particular, the locality of $C$ is obtained from the locality of the entries in any row of $\textbf{c}\in C$.



To compute the locality of an $[m\ell,k,d]$ QC code $C$ using the associated cyclic code $D$ of $C$, we assume that the codewords $\textbf{c}$ are described by the trace-formula representation provided in~\cite{cem}, which naturally may also be represented  using the array notation in~\eqref{codeword}. We then see that the columns of $\textbf{c}$ (represented as an array in~(\ref{codeword})), obtained by this trace-formula representation of $C$ from the unique finite extension field, are codewords of the cyclic code $D$ associated to $C$. 
%
%


The dual code $D^{\perp}$, as defined in~\eqref{bz},
 has  basic zero set 
\begin{equation}\label{Def_BZSet} BZ(D^{\perp})=\{\beta ^{- u_{a_1}}, \beta^{- u_{a_2}}, \ldots, \beta^{- u_{a_h}} \}
\end{equation}
where 
\begin{equation}\label{matrixcodeword}
\{\beta^{ u_{a_1}}, \beta^{ u_{a_2}}, \ldots, \beta^{ u_{a_h}} \} \subset \{\beta^{ u_{1}}, \beta^{ u_{2}}, \ldots, \beta^{ u_{s+1}}\} .
\end{equation}
As $C$ is a nontrivial code, we note that $D^{\perp} \neq \mathbb{F}_q ^m$ and from Remark~\ref{rem1}, we know that the locality of $D$ is $d(D^{\perp}) -1$. Thus, for any codeword $\textbf{c}$ of $C$, we can flatten the array in~(\ref{codeword}) by the columns and rewrite $\textbf{c}$  as a row vector of length $m\ell$. Moreover, if the codewords $\textbf{d}_0,...,\textbf{d}_{\ell -1} \in D$ represent the $m$-length blocks of $\textbf{c}$, we have that
\begin{equation}
\textbf{c} = (\underbrace{c_{00}, c_{10}, ... ,c_{m-1,0}}_{\textbf{d}_0} ,\ldots , \underbrace{c_{0, l-1}, c_{1,l-1}, ... ,c_{m-1,l-1}}_{\textbf{d}_{\ell -1} }) \in \mathbb{F}_q ^{ml}.
\end{equation}
Using this representation, we note that each coordinate of $\textbf{c}$ belongs to an $m$-length block $\textbf{d}_i$, $0\leq i\leq \ell -1$, where each coordinate in this $i$-th block has the same locality, namely, $\mathrm{loc}(c_{ij})=d(D^{\perp}) - 1.$

The following result can be seen as a restatement of Theorem~\ref{locRB} and Corollary~\ref{corlocRB}.

\begin{proposition}
\label{localityofqc}
Let $C$ be an $[m\ell,k,d]$ QC code with locality $r$. If $D$ is the cyclic code associated to $C$, then $r\leq d(D^{\perp}) -1$.
\end{proposition}


Proposition~\ref{localityofqc} is useful for a QC LRC $C$ where either $m$ is small or $C$ has a large number of generators. In the most extreme case, since $D^{\perp}$ is a cyclic code of length $m$, then $r\leq d(D^{\perp}) -1\leq m-1$. As this upper bound for the locality only depends on $m$, we have a straightforward upper bound that involves no computation, especially when compared to the work in~\cite{RB}. For example, consider the QC $[m_1\ell_1=4\cdot 4, k_1=8]$ code $C_1$ in~\cite[Example 1]{RB}, the QC $[m_2\ell_2=4 \cdot 4, k_2=4]$ code $C_2$ in~\cite[Example 2]{RB}, and the QC $[m_3\ell_3=5 \cdot 3, k_3=5]$ code $C_3$ in \cite[Example 3]{RB}. By~\cite[Theorem 1]{RB}, these QC codes have localities $r_1 \leq 4$, $r_2 \leq 3$, and $r_3 =3$, respectively. However Proposition~\ref{localityofqc} directly tells us that $r_1$ and $r_2$ are upper bounded by $m_1 -1 = m_2 -1  = 3$, and that $r_3$ is upper bounded by $m_3  -1 = 4$.

\section{Constructing Optimal QC LRCs using Constituent Codes} \label{examples}

\par 
To construct an (almost) optimal QC LRC $C$, we measure how far its lower bound  $d_{GO}(C)$ \eqref{GO_Bound} of the minimum distance is from its respective Singleton-type upper bound $d_S(C)$ \eqref{singletontypebound}. To have an almost optimal code, we aim for the difference between $d_{GO}(C)$ and $d_S(C)$  to be equal to $1$, and to have an optimal code, we aim for the equality $d_{GO} (C) = d_{S} (C)$.

In this section, we utilize the constituent codes of a given QC LRC to construct another QC LRC whose lower and upper bounds of the minimum distance meet, namely: given a $[n, k, d]$ QC LRC $C$, we construct a respective $[n', k', d']$ QC LRC $C'$ which is longer and larger than $C$ (i.e., $n' > n$ and $k' > k$) whose respective constituent codes each have the same minimum distance of the constituent codes of $C$. Consequently,  $C'$ inherits the lower bound for its minimum distance---that is, $d_{GO}(C) = d_{GO}(C')$, and the locality of $C'$ is also the same as $C$. With $C'$ having a larger code length, larger dimension, and the same locality as $C$, its respective Singleton-type bound  becomes reduced when compared to $d_S(C)$, i.e., 
\begin{equation}
d_S (C') - d_{GO} (C') = d_S (C') - d_{GO} (C) \leq d_S (C) - d_{GO} (C).
\end{equation}

Therefore, in this section we are interested in constructing a QC LRC code $C'$ that is (almost) optimal. The following example is almost optimal code which uses constituent codes to illustrate the difference, or the ``gap", between the Singleton-type upper bound and  the G\"uneri-\"Ozbudak lower bound for the minimum distance of a QC LRC.

\begin{example}\label{example_optimal}
Let $q=2$, $\ell=3$, $m=7$, and $\beta$ be a primitive $7$-th root of unity. Assume the following decomposition of $R^3 = \left(\frac{\mathbb{F}_2 [x]}{\langle x^7 -1 \rangle}\right)^3$ into its finite extension fields provided by Chinese Remainder Theorem:
\begin{eqnarray*}\label{CRT_decomposition_R^3}
\left(\frac{\mathbb{F}_2 [x]}{\langle x^7 -1 \rangle}\right)^3 &\cong&  \left(\frac{\mathbb{F}_2 [x]}{\langle x^3 +x +1 \rangle} \right)^3 \oplus\left(\frac{\mathbb{F}_2 [x]}{\langle x^3 +x^2 +1 \rangle}    \right)^3\oplus\left(\frac{\mathbb{F}_2 [x]}{\langle x -1 \rangle}\right)^3  \nonumber \\
&\cong &  \mathbb{F}_2 ^3(\beta) \oplus \mathbb{F}_2 ^3 (\beta^3 ) \oplus  \mathbb{F}_2 ^3.  
\end{eqnarray*}

Let $C$ be a QC code in $\mathbb{F}_2 ^{21}$ whose corresponding constituent codes are
\begin{align*}C_1  &= \langle (\beta^2 + \beta +1, \beta^2 + \beta +1, \beta^2 +1), (0, \beta^2 +1, \beta^2 +1)\rangle \subseteq \mathbb{F}_2 ^3 (\beta),\\
C_2  &= \mathbb{F}_2 ^3 (\beta^3) = \mathbb{F}_8 ^3 , \\
C_3  &= \{(0,0,0)\}\subseteq \mathbb{F}_2 ^3.
\end{align*}
As $C_1$ is a $[3,2,2]$ code and $C_2$ is a $[3,3,1]$ code, they are both optimal codes over $\mathbb{F}_8$. From Lemma~\ref{dimmension_lemma}, the dimension of $C$ is $15$. 
%
Furthermore, as $BZ(D_1 ^{\perp}) = \{\beta^{-1}\}$, $BZ(D_2 ^{\perp}) = \{\beta^{-3}\}$, and $BZ(D_{1,2} ^{\perp}) = \{\beta^{-1},\beta^{-3}\}$, the $7$-length cyclic code $D=D_{1,2}$ associated to C and its respective cyclic subcodes are
\begin{equation}
D_1 = \langle x^4 +x^2 +x +1 \rangle , D_2 = \langle x^4 +x^3 +x^2 +1 \rangle, \mbox{ and } D_{1,2} = \langle x +1 \rangle,  
\end{equation}
where $d(D_1) = d(D_2) = 4$, and $d(D_{1,2})=2.$ By Theorem~\ref{GO_Bound_General}, we have
\begin{equation}\label{example_guneri}
d\geq \min\{d(C_2) d(D_2), [d(C_1) - d(C_2)]d(D_1) + d(C_2) d(D_{1,2})\} = 4,
\end{equation}
and $C$ is a $[21,15,d\geq 4]$ QC code. Moreover, the dual of $D_{1,2}$ is
\begin{equation}
D_{1,2} ^{\perp}= \langle x^6 +x^5 +x^4 +x^3 +x^2 +x +1\rangle    ,
\end{equation}
namely, it is a [7,1,7] cyclic code (repetition code). From Proposition~\ref{localityofqc}, we conclude that $r\leq 6$ and
\begin{equation}
5= 21 - 15 -\left\lceil \frac{15}{6} \right\rceil +2 = d_S(C) \geq d(C) \geq 4.
\end{equation}
%
Therefore, $C$ is, at least, an almost optimal QC LRC.
\end{example}

With inspiration from Example~\ref{example_optimal}, we  construct a family of optimal QC LRCs by adding additional hypotheses on the constituent codes, namely by increasing their dimension and code length while preserving their minimum distances. The code construction is as follows: Let $C \subset \mathbb{F}_q ^{m\ell}$  be an $[m\ell, k, d]$  QC LRC of index $\ell$ and locality $r$. Factor $x^m-1$ into $s+1$ irreducible factors as in~(\ref{factorization}) so that we have the isomorphism described in (\ref{eq:crtiso}). After re-ordering the factors of $x^m-1$, let $C_i \subset \mathbb{F}_q ^{\ell} \left(\beta^{u_i}\right)\cong \left(\mathbb{F}_q[x]/\langle b_i(x) \rangle\right)^\ell$
be the $[\ell, k_i , d_i]$ non-zero constituent codes of $C$, where $i=1,...,h \leq s+1$, whose minimum distances satisfy (\ref{descending_minimum_distance}), i.e.,
\begin{equation}
 C \cong C_{1} \oplus C_{2} \oplus \cdots \oplus C_{h},\mbox{ where } d(C_1)\geq d(C_2) \geq \ldots \geq d(C_h) .    
\end{equation}

For each $i=1,2,...,h$, define a $[\ell+j, k_i+j, d_i]$ linear code $C_{i}^j$ for $j \in J$, where $J \subseteq \mathbb{Z}^{\geq 0}$ is the set of integers that ensures the existence of the linear code $C_{i}^j$  with the specified parameters. 
\begin{remark}
Note that $|J| \leq \infty$. Furthermore, if $C$ has a zero constituent code $C_i = \{(0,0,...,0)\}$, for some $i$ such that $h+1 \leq i\leq s+1$, then $C_i ^j$ is still a zero constituent code, for any $j\in J.$    
\end{remark}
Let $C^{j} \subset \mathbb{F}_q ^{m(\ell+j)}$ be the respective QC LRC with constituent codes $C_{i}^j$ where $i=1,2,...,s+1$ and $j \in J$. In other words, via the inverse of the Chinese Remainder Theorem, $\iota$, and $\phi$ defined in Section \ref{preliminaries},

\begin{equation}\label{definecj}
C^j = \phi^{-1} \left( \iota^{-1}( C_{1}^j \oplus C_{2}^j + \cdots + C_{s+1}^j)\right) .
\end{equation}

We note that $C^j$ has code length $m(\ell+j)$ and index $\ell+j$. By Lemma~\ref{dimmension_lemma}, $C^j$ has dimension $ \sum_{i=1}^{s+1}(k_i+j)\deg(b_i(x))$.  Moreover, for any $j\in J$, the locality of $C^j$ is the same as the locality of $C$ because the locality is obtained from the cyclic code $D$ associated to $C$, which depends on exclusively of $m$ and the factorization of $x^m -1$. Since  $m$ is unchanged in every $C^j$, the locality $r$ of $C^j$ is also unchanged for all $j \in J$.

\begin{remark}\label{same_GO_bound}
It is worth highlighting that $d_{GO}(C)=d_{GO}(C^{j})$ for any $j \in J$ as the G\"uneri-\"Ozbudak lower bound (Theorem~\ref{GO_Bound_General}) depends exclusively on the minimum distances of the constituent codes of $C$, the minimum distance of the cyclic code $D$ associated to $C$, and the minimum distances of the respective cyclic subcodes of $D$. Moreover, as $d(C_i ^j) = d(C_i)$ for any $1\leq i\leq s+1$ and $j\in J$, and the parameter $m$ is unchanged for each $j$, this guarantees that $C$ and $C^j$ have the same associated cyclic code $D$. This ensures that $d_{GO}(C)=d_{GO}(C^{j})$ for any $j \in J$.
\end{remark}

From now on, we assume that $0\in J$ and $C^0 = C.$ The next example shows how we can use this construction to present a family of QC LRCs which are at least almost optimal, for $j \in J$.

\begin{example}\label{first_cj_example} Let $C^j$ be an $[7\ell, k, d]$ QC code in $\mathbb{F}_{2}^{7\ell}$ of locality $r$ and index $\ell$, where $\ell = 3+ j$ for $j \geq 0$. Consider the following decomposition of $R^\ell = \left( \frac{\mathbb{F}_{2}[x]}{\langle x^7-1 \rangle} \right)^\ell$ obtained by the Chinese Remainder Theorem:
\begin{eqnarray*}\label{CRT_decomposition_R^3}
\left(\frac{\mathbb{F}_2 [x]}{\langle x^7 -1 \rangle}\right)^\ell &\cong&  \left(\frac{\mathbb{F}_2 [x]}{\langle x^3 +x +1 \rangle} \right)^\ell \oplus\left(\frac{\mathbb{F}_2 [x]}{\langle x^3 +x^2 +1 \rangle}    \right)^\ell\oplus\left(\frac{\mathbb{F}_2 [x]}{\langle x -1 \rangle}\right)^\ell  \nonumber \\
&\cong &  \mathbb{F}_2 ^\ell(\beta) \oplus \mathbb{F}_2 ^\ell (\beta^3 ) \oplus  \mathbb{F}_2 ^\ell.  
\end{eqnarray*}

Let $C_1 ^j$, $C_{2}^j$ and $C_3$ be the constituent codes of $C^j$ where
$C_{1}^j= \langle (1,1,\ldots,1)\rangle ^{\perp}$,  $C_2 ^j = \mathbb{F}_{2}^\ell(\beta^3 )$,  and $C_3=\{(0,...,0)\} \subseteq \mathbb{F}_{2}^\ell$. In particular, $C_1 ^j$ is a $[3+j, 2 +j , 2]$ linear code in $\mathbb{F}_{2}^\ell(\beta)$ and $C_2 ^j$ is a $[3+j, 3+j, 1]$ linear code in $\mathbb{F}_2 ^\ell (\beta^3 ) $. 
%
%
Note that both $C_1 ^j$ and $C_2 ^j$ exist for any $j\geq 0$ and are optimal linear codes. From Lemma~\ref{dimmension_lemma}, we have that the dimension of $C^j$ is $k=(2+j)\cdot 3+(3+j)\cdot 3.$ 
From Theorem~\ref{GO_Bound_General}, $d_{GO} (C^j) = 4$. As seen in Example~\ref{example_optimal}, since $D_{1,2}^{\perp}=\langle (x^3+x+1)(x^3+x^2+1) \rangle$, the locality of $C$ is $r\leq 6$. Thus,
\begin{eqnarray*}\label{ineqoptimalfamily}
d_S (C^j)& = &7(3+j) - [3(2+j) + 3(3+j)] -\left\lceil\frac{3(2+j) + 3(3+j)}{r} \right\rceil +2 \nonumber \\
&\leq & 8 +j -\left\lceil\frac{15 + 6j}{6} \right\rceil
\nonumber  \\
&=& 5
\end{eqnarray*}
  
Thus, $4 = d_{\text{GO}}(C^j) \leq d(C^j) \leq d_{S}(C^j)=5$ which shows that $C^j$ is at least an almost optimal code for any $j \geq 0$. 

\end{example}

With each $j \in J$, this set up using constituent codes allows for the construction of a QC LRC with increased code length and dimension. However, it is worth noting that the code $C^j$ in Example~\ref{first_cj_example} has a Singleton-type bound that is independent of $j$, which creates a constant upper bound on its minimum distance. This is a special case. In general, assuming that $J$ contains other numbers than zero, the Singleton-type bound on the minimum distance of $C^j$ will depend on $j$ which creates a descending chain of Singleton-type bounds for the minimum distance as stated in the following theorem. 
\begin{theorem}\label{mainth}
Let $C\subset \mathbb{F}_q ^{m\ell}$ be a QC code of index $\ell$ and locality $r$. Let $C^j \subset \mathbb{F}_{q}^{m(\ell+j)}$ be the  QC LRC described in~(\ref{definecj}) with $h$ nontrivial constitutent codes $C_i^j$ with parameters $[\ell+j, k_i+j, d_i]$ for $j \in J$ and $i=1,2,...,h \leq s+1$. Let $b_i = \deg(b_i (x))$, for each $1\leq i \leq h$.
If
\begin{equation}\label{main_theorem_inequality}
m + 1 - \left(\sum_{i=1} ^{h} b_i + \left\lceil \frac{1}{r}\sum_{i=1} ^{h}b_i  \right\rceil \right)\leq 0,
\end{equation}
then for any $j \in J$,
\begin{equation}\label{main_theorem_descendingchain}
d_S (C^{j+1}) \leq d_S (C^{j}).
\end{equation}
Moreover, if $d_S(C^{j_0}) = d_{\text{GO}}(C^{j_0})$ for some $j_0 \in J$, then each QC LRC $C^{j}$ is optimal, for all $j\geq j_0$ in $J$.
\end{theorem}
%

\begin{proof}
By  Lemma~\ref{dimmension_lemma}, the dimension of $C^j$ is $\sum_{i=1}^{h} (k_i+j)b_i$, for each $j \in J$. Thus, for any $j \in J$, the Singleton-type bound in (\ref{singletontypebound}) for $C^{j+1}$ is 
\begin{eqnarray}
d_S \left(C^{j +1}\right)&=& m(l+j +1) - \sum_{i=1} ^{h} \left(k_i + j +1\right)b_i -\left\lceil \frac{\sum_{i=1} ^{h} \left(k_i + j +1\right)b_i}{r}\right\rceil +2 \nonumber \\  
&=& m(l+j) - \sum_{i=1} ^{h} (k_i + j)b_i - \left\lceil \frac{ \sum_{i=1} ^{h} (k_i + j)b_i}{r}\right\rceil + \left\lceil \frac{ \sum_{i=1} ^{h} (k_i + j)b_i}{r}\right\rceil  \nonumber\\
&& + m -\sum_{i=1} ^{h} b_i \nonumber -\left\lceil \frac{\sum_{i=1} ^{h} \left(k_i + (j+1)\right)b_i}{r}\right\rceil +2\nonumber \\
&=& \left( m(l+j) - \sum_{i=1} ^{h} (k_i + j)b_i - \left\lceil \frac{\sum_{i=1} ^{h} (k_i + j)b_i}{r}\right\rceil  +2 \right) \nonumber \\
&& +\left\lceil \frac{\sum_{i=1} ^{h} (k_i + j)b_i}{r}\right\rceil - \left\lceil \frac{\sum_{i=1} ^{h} \left(k_i + (j+1)\right)b_i}{r}\right\rceil  +m -\sum_{i=1} ^{h} b_i 
 \nonumber\\
&=& d_S (C^j ) + \left\lceil \frac{\sum_{i=1} ^{h} (k_i + j)b_i}{r}\right\rceil - \left\lceil \frac{\sum_{i=1} ^{h} \left(k_i + (j+1)\right)b_i}{r}\right\rceil  +m -\sum_{i=1} ^{h} b_i.
 \nonumber
 \end{eqnarray}
For any $x,y \in \mathbb{Q}$, 
 $- (\left\lceil  x \right\rceil + \left\lceil  y \right\rceil -1 )\geq -\left\lceil  x + y\right\rceil $. Thus, we have
\begin{eqnarray}\nonumber
 d_{S}(C^{j+1}) &\leq& d_S (C^j ) + \left\lceil \frac{\sum_{i=1} ^{h} (k_i + j)b_i}{r}\right\rceil - \left\lceil \frac{\sum_{i=1} ^{h} \left(k_i + j \right)b_i}{r}   \right\rceil - \left\lceil \frac{\sum_{i=1} ^{h}b_i}{r}\right\rceil \\ \nonumber
 && +1 +m -\sum_{i=1} ^{h} b_i \\
\nonumber
\\
&=& d_S (C^j ) +  m + 1 - \left(\sum_{i=1} ^{h} b_i + \left\lceil \frac{\sum_{i=1} ^{h}b_i}{r}   \right\rceil \right) 
 \nonumber\\
&\leq& d_S (C^j ),\nonumber
\end{eqnarray}
where the last inequality holds true due to the assumption in (\ref{main_theorem_inequality}). 
Now assume that there exists $j_0 \in J$ such that $d_{S}(C^{j_0})=d_{GO}(C^{j_0})$. Then from~\eqref{main_theorem_descendingchain}, we have
\begin{equation}\label{descendingchain}
d_S(C) \geq \cdots \geq d_{S}(C^{j_0}) = d_{GO}(C^{j_0}) \geq d_{S}(C^{j_0+1}) \geq \cdots \geq d_{S}(C^{|J|}).
\end{equation}
As the lower bound and upper bound for the miniminum distance of $C^{j_0}$ are equal, $C^{j_0}$ is optimal. Also, by Remark \ref{same_GO_bound}, $C^{j}$ is optimal for all $j > j_0$ in $J$.
\end{proof}

\begin{remark} We note that the set $J$ only includes values of $j$ where the Singleton-bound $d_S$ is positive. Moreover, if $|J|=\infty$, then the descending chain in \eqref{descendingchain} will be infinite.
\end{remark}

\begin{corollary}  Let $C^j \subset \mathbb{F}_{q}^{m(\ell+j)}$ be the  QC LRC described in~(\ref{definecj}). If $C^j$ has $s+1$ nontrivial constitutent codes with parameters $[\ell+j, k_i+j, d_i]$ for each $j \in J$, then 
$d_S (C^{j+1}) \leq d_S (C^{j})$.  
\end{corollary}

\begin{proof} If $C^j$ has no zero constituent codes, then $m =\sum_{i=1} ^{s+1} b_i$ and
$$d_S(C^{j+1}) \leq d_{S}(C^j) - \left\lceil \dfrac{m}{r}\right\rceil+1.$$
As $r\leq m-1 $ by Proposition~\ref{localityofqc}, we have  
$d_S(C^{j+1}) \leq d_{S}(C^j).$
\end{proof}


\begin{corollary}
Under hypothesis of Theorem~\ref{mainth}, if $C$ is an optimal QC LRC  then $C^j$ is also an optimal QC code for any $j\in J$.
\end{corollary}

In the next example, we use our code construction to prove the existence of a QC LRC with increased code length and dimension which are optimal for a certain value of $j \in J$.  In the next example, we note that we change the order of our factorization so that the first constituent code in our decomposition is described over $\mathbb{F}_5$, and inequality~\eqref{descending_minimum_distance} is satisfied.

\begin{example}\label{optimal_example}
Let $q=5,\,m=11$, and $\ell=7$, $\beta \in \mathbb{F}_{5^5}$ be an $11-$primitive root of unity. Thus, by Chinese Remainder Theorem, we have 
\begin{eqnarray*}
\left(\frac{\mathbb{F}_5 [x]}{\langle x^{11} -1\rangle}\right)^7 &\cong& \left(\frac{\mathbb{F}_5 [x]}{\langle x+4\rangle}\right)^7 \oplus \left(\frac{\mathbb{F}_5 [x]}{\langle x^5 +2x^4 +4x^3 + x^2 + x +4 \rangle}\right)^7\nonumber\\
&\oplus& \left(\frac{\mathbb{F}_5 [x]}{\langle x^5 +4x^4 +4x^3 + x^2 + 3x +4 \rangle}\right)^7\nonumber \\ 
&\cong& \mathbb{F}_5 ^7 \oplus \mathbb{F}_{625} ^7 \oplus 
\mathbb{F}_{625} ^7 = \mathbb{F}_5 ^7  (1 ) \oplus \mathbb{F}_{5}^7 (\beta^2)  \oplus \mathbb{F}_{5}^7 (\beta )  .
\end{eqnarray*}
Since $$BZ(D^{\perp})= \{1, \beta_2 ^{-1}  =\beta^9, \beta ^{-1} = \beta^{10}  \},$$ we have
$$
\begin{array}{ccccc}
D_1= \langle x^{10}+ x^9 + x^8 +x^7 +x^6  +x^5 +x^4 + x^3 + x^2 +x +1\rangle &\Rightarrow& d(D_1) &=& 11\\
D_2= \langle  x^6 +3x^5 +2x^3 +2x^2 + x+1  \rangle &\Rightarrow& d(D_2) &=& 6\\
D_3= \langle  x^6 +x^5 +2x^4 +2x^3 +3x +1 \rangle &\Rightarrow& d(D_3) &=& 6\\
D_{1,2}= \langle x^5+ 4x^4 + 4x^3 + x^2 + 3x +4\rangle &\Rightarrow& d(D_{1,2}) &=& 5\\
D_{1,3}= \langle x^5+ 2x^4 + 4x^3 + x^2 + x+4\rangle &\Rightarrow& d(D_{1,3}) &=& 5\\
D_{2,3}= \langle x +4 \rangle &\Rightarrow& d(D_{2,3}) &=& 2 \\
D_{1,2,3}= \langle 1 \rangle = \mathbb{F}_5 ^{11}&\Rightarrow& d(D_{1,2,3}) &=& 1
\end{array}
$$
where these minimum distances were computed using \texttt{SageMath}. 

Let $C$ be a QC LRC in $\left(\frac{\mathbb{F}_5 [x]}{\langle x^{11} -1\rangle}\right)^7$ and let $C_1$, $C_2$ and $C_3$ be the constituent codes of $C$, where $C_1\subset \mathbb{F}_5 ^7 $ is a $[7,3,4]$ linear code, $C_2 \subset \mathbb{F}_5  ^7 (\beta^2 )$ is a $[7,4,3]$ linear code, and is $C_3 \subset \mathbb{F}_5 ^7(\beta) $ a $[7,5,2]$ linear code. Note that the minimum distances of $C_1, C_2$ and $C_3$ satisfy inequality~\eqref{descending_minimum_distance}. The existence and how to construct such codes may can be found in~\cite{table}. Using these constituent codes and the information in~\cite{table} (which ensures the existence of a set $J$), we define the constituent codes $C_1 ^j$ , $C_2 ^j$, and $C_3 ^j$ whose parameters are $[7+j, 3+j,4]$, $[7+j, 4+j,3]$, and $[7+j, 5+j,2]$, respectively, where the minimum distance of the constituent codes is preserved for all $j\geq 0$. The respective QC code
\begin{equation}
C^j = \phi^{-1} \left( \iota^{-1}( C_1 ^j \oplus C_2 ^j  \oplus C_3 ^j )\right)   
\end{equation}
has length $ 11(7+j) = 77+ 11j$. According to Lemma~\ref{dimmension_lemma} and Proposition~\ref{localityofqc}, $C^j$ has dimension $k = 3+j +5(4+j) +5(5+j) = 48 + 11j$ and its locality $r$ is upper bounded by $10$. In addition,  the Singleton-bound of $C^j$ is
\begin{eqnarray}\label{singboundex}
d_S (C^j ) &=& 77+ 11j - (48 +11j) - \left\lceil \frac{48 +11j}{r}\right\rceil +2\nonumber \\
 &\leq& 77+ 11j - (48 +11j) - \left\lceil \frac{48 +11j}{10}\right\rceil +2\nonumber \\
&=& 31 - \left\lceil \frac{48 +11j}{10}\right\rceil
\end{eqnarray}
where $j \in J=\{1,2,3,...,22\}$. Note that for $j \in J$, inequality ~\eqref{main_theorem_inequality} is satisfied and consequently ~\eqref{main_theorem_descendingchain} in Theorem~\ref{mainth}. On the other hand, the upper bound for the minimum distance from Theorem~\ref{GO_Bound_General}  yields
\begin{equation}
d_{GO} (C) = d_{GO} (C^j) =  \min\{R_3 , R_{2,3} , R_{123} \}=10, 
\end{equation}
where 
\begin{align*}
R_3 &= d(C^j_3) d(D_3) =12;\\
R_{2,3}&=(d(C^j_2) - d(C^j_3) )d(D_2) + d(C^j_3) d(D_{2,3}) =10;\\
R_{1,2,3} &= (d(C^j_1) - d(C^j_2) )d(D_1) + (d(C^j_2) -d(C^j_3) )d(D_{1,3}) + d(C^j_3) d(D_{1,2,3}) = 18.
\end{align*}
When $j = 14$, the bound in~\eqref{singboundex} becomes $31 - \left\lceil \frac{48 +11j}{10}\right\rceil = 10$. Hence, we have
\begin{equation}
10 = d_{GO} (C^{14}) \leq d(C^{14}) \leq d_S (C^{14}) \leq 10 . 
\end{equation}
which shows that $C^{14}$ is an \textit{optimal} $[231, 202, 10]$ QC LRC with index $21$ and locality $r \leq 10$. 
\end{example}

We note the following differences between the last two examples in this section. In Example \ref{first_cj_example}, the Singleton-type bound for $C^j$ is independent of $j$, for each $j \in J$. In Example \ref{optimal_example}, the upper bound on the minimum distance of $C^j$ depends on $j$. This allows us to find a specific $j \in J$ so that $d_S (C^j) = d_{GO} (C^j) = d_{GO} (C)$, which in turn shows that $C^j$ is optimal and has an increased code length and dimension when compared to $C$.

\section{Conclusion and Future Directions}

In this paper, we have presented a construction of optimal QC LRCs in $\mathbb{F}_q ^{m\ell}$ using constituent codes studied in ~\cite{cem,lingsole}. We first presented an upper bound for the locality of a QC LRC using different approach from the bound presented in~\cite{RB}. Using this upper bound for the locality, our main result relies on when the Singleton-type bound for the minimum distance for a QC LRC meets the G\"uneri-\"Ozbudak´s lower bound for the minimum distance. In particular, such a lower bound depends on the minimum distance of the constituent codes of the QC LRC,  the cyclic code associated to the QC LRC, and the cyclic subcodes of the associated cyclic code.  As the minimum distance of the associated cyclic code depends on $m$, we analyze when it is possible to obtain new constituent codes by increasing their code length and dimension while preserving their original minimum distances. Such an increase in code length and dimension results in a respective QC LRC whose respective Singleton-type bound will be reduced to (nearly) attain the G\"uneri-\"Ozbudak´s lower bound. This construction not only produces almost optimal and optimal QC LRCs, but codes with an increased dimension and code length.

For future directions of this work, we intend to further study constituent codes of QC LRCs to determine conditions on the constituent codes' parameters that will guarantee that their respective QC LRC is optimal. Additionally, we plan to use constituent codes and self-orthogonality conditions to construct good quantum error-correcting codes and determine their parameters.

\bibliographystyle{plain}
\bibliography{QC_LRC}

\end{document}